\documentclass[reqno,11pt]{amsart}
\usepackage{setspace,tikz,xcolor,mathrsfs,listings,multicol,amssymb,amsfonts,amsthm}
\usepackage{pgfplots}
\pgfplotsset{compat=1.15}
\usepackage{rotating}
\usepackage{tkz-base}
\usepackage{tkz-euclide}
\usepackage{tikz-cd}
\usepackage{comment}
\usepackage[margin=1in,includefoot,footskip=30pt]{geometry}
\usepackage{enumerate}
\usepackage{booktabs}
\usepackage[all,cmtip]{xy}
\usetikzlibrary{arrows,matrix,math}

\tikzset{tab/.style={matrix of math nodes,column sep=-.35, row sep=-.35,text height=7pt,text width=7pt,align=center,inner sep=2,font=\footnotesize}}

\usepackage{hyperref}
\hypersetup{hidelinks,
 pdftitle={Christoffel transformations of symmetric d-connections and gap probability recurrences},
 pdfauthor={Anton Selemenchuk}}

\IfFileExists{mathc10.tfm}{
\DeclareFontFamily{U}{mathc}{}
\DeclareFontShape{U}{mathc}{m}{it}%
{<->s*[1.03] mathc10}{}
\DeclareMathAlphabet{\mathcal}{U}{mathc}{m}{it}
}{}

\definecolor{darkred}{rgb}{0.7,0,0} 

\definecolor{UQgold}{RGB}{196, 158, 54} 
\definecolor{UQpurple}{RGB}{73, 7, 94} 

\usepackage{listings}
\lstdefinelanguage{Sage}[]{Python}
{morekeywords={False,sage,True},sensitive=true}
\definecolor{dblackcolor}{rgb}{0.0,0.0,0.0}
\definecolor{dbluecolor}{rgb}{0.01,0.02,0.7}
\definecolor{dgreencolor}{rgb}{0.2,0.4,0.0}
\definecolor{dgraycolor}{rgb}{0.30,0.3,0.30}

\theoremstyle{plain}
\newtheorem{thm}{Theorem}[section]
\newtheorem{lemma}[thm]{Lemma}

\newtheorem{prop}[thm]{Proposition}
\newtheorem{cor}[thm]{Corollary}
\theoremstyle{definition}
\newtheorem{dfn}[thm]{Definition}

\newtheorem{remark}[thm]{Remark}

\numberwithin{equation}{section}
\begin{document}
\title[Christoffel transformations and gap recurrences]
{Christoffel transformations of symmetric $d$-connections
and gap probability recurrences}
 \author[A.~Selemenchuk]{Anton Selemenchuk}
   \address[A.~Selemenchuk] {Faculty of Mathematics and Computer Science,
    St.\ Petersburg State University, University Embankment, 7/9, St.\
    Petersburg, Russia, 199034}
  \email{selemenchyk@icloud.com}
\subjclass[2020]{Primary 14D20;
Secondary 14H60, 39A36, 33C45, 60G55}

\begin{abstract}
We study the Christoffel transformation $w(x)\mapsto x^2w(x)$
of an even discrete weight. For the associated orthogonal polynomial
ensembles on squared nodes $y=x^2$, ratios of moment determinants
relate the largest-particle distributions and transfer exact gap
probability recurrences between the ensembles. For positive weights
satisfying the stated moment and support assumptions, this comparison
requires no Pearson condition.
The transformation is realized by elementary bundle modifications
at zero and infinity on $\mathbb P^1$. 

When the original weight has
Pearson ratio $d(-z-1)/d(z)$ with $\mathrm{deg}\;  d=1,2$, these modifications
conjugate cutoff dynamics on common stable loci of symmetric
$d$-connections under the stated geometric hypotheses.
The transformed recurrence coefficients are expressed explicitly
as rational functions of spectral coordinates.
On the unit lattice, a local recurrence formula applies to rational
symmetric Pearson ratios of arbitrary degree, including outside
these charts, whenever its local hypotheses hold.\end{abstract}
\maketitle
\tableofcontents
\section{Introduction}
The distribution function of the largest particle in a discrete
orthogonal polynomial ensemble is a gap probability: the probability
that no particle lies above a prescribed cutoff.
Such ensembles arise in models of random tilings and random Young
diagrams; see \cite{Okounkov2001,BKMM2007}.
Their gap probabilities connect probability theory, orthogonal
polynomials and integrable systems.
Multiplying the weight by a polynomial produces a related ensemble.
We study how exact recurrences for gap probabilities transform under
this change, keeping track of the normalization factors and the
associated geometric dynamics.
We focus on the Christoffel transformation
$w(x)\mapsto\widehat w(x)=x^2w(x)$ of an even discrete weight. Multiplying the weight by a polynomial produces related ensembles.
The problem is to transfer exact gap probability recurrences under this
change while controlling their normalization and geometric dynamics.
We address this problem for the Christoffel transformation
$w(x)\mapsto\widehat w(x)=x^2w(x)$ of an even discrete weight.

One motivation comes from symplectic skew Howe duality
\cite{Howe1989,NazarovPostnovaScrimshaw2024}. For the dual pair
$(\mathrm{Sp}_{2n},\mathrm{Sp}_{2\ell})$, the dimensions of the irreducible
summands define a probability measure on diagrams
$\lambda\subseteq\ell^n$. In the coordinates
$a_i=\lambda_i+n-i+1$ and $N=n+\ell$, this measure takes the form
\[
 \mathbb P(a_1,\ldots,a_n)\ \propto\
 \prod_{i<j}(a_i^2-a_j^2)^2
 \prod_{i=1}^n a_i^2\binom{2N}{N+a_i},
 \qquad N\ge a_1>\cdots>a_n\ge1;
\]
see \cite[Section~2]{NNS2026}. Thus the factor $a_i^2$ is precisely a
Christoffel deformation of the symmetric Krawtchouk weight. Moreover,
$\mathbb P(\max_i a_i\le s)=\mathbb P(\lambda_1\le s-n)$, so a gap
probability gives the distribution of the first row of the diagram.
The ensemble is determinantal: the symplectic group enters through
representation theory, whereas the squared Vandermonde corresponds
to an orthogonal polynomial ensemble with $\beta=2$.

The comparison of projection kernels in \cite{NNS2026} transfers
local asymptotic results between ensembles related by a Christoffel
transformation. Here we study the exact dependence of gap probabilities
on the cutoff for a fixed number of particles. We seek a geometric
description that also applies to deformations of classical weights.
Indeed, if $w(j+1)/w(j)=\varrho(j)$, then
\[
 \frac{\widehat w(j+1)}{\widehat w(j)}
 =\frac{(j+1)^2}{j^2}\varrho(j),\qquad j,j+1\ne0.
\]
Even for the Krawtchouk ratio $\varrho(z)=(N-z)/(N+z+1)$, the
transformed ratio has numerator and denominator of degree three.
Bundle modifications transfer the original spectral data without
imposing the original degree bound on the new weight. The construction
therefore applies to deformed, potentially nonclassical ensembles.

The continuous theory provides the analytic background. Fokas, Its and
Kitaev \cite{FokasItsKitaev1992} characterized orthogonal polynomials by
a matrix Riemann--Hilbert problem. The nonlinear steepest descent method
of Deift and Zhou \cite{DeiftZhou1993} provides a systematic approach to
the asymptotic analysis of oscillatory Riemann--Hilbert problems; see Deift's monograph
\cite{Deift1999} for applications to orthogonal polynomials and random matrices.

Tracy and Widom \cite{TracyWidom1994} derived differential equations for
the Fredholm determinants representing gap probabilities. In particular,
the largest-eigenvalue distributions for the Hermite and Laguerre unitary
ensembles of finite matrix size admit representations through
Painlev\'e~IV and~V transcendents, respectively. In the soft-edge scaling limit,
the Airy-kernel Fredholm determinant gives
the Tracy--Widom distribution $F_2$, which is expressed through the solution of Painlev\'e~II \cite{TracyWidomAiry1994}.

In the discrete setting, Borodin and Boyarchenko \cite{Borodin_2003}
derived recurrences for first-particle distributions under a rational
Pearson assumption. Classical specializations give difference Painlev\'e
equations. Their method already extends beyond
classical weights; general discrete weights are also central to the
asymptotic theory in \cite{BKMM2007}.
Arinkin and Borodin \cite{arinkin2007,ArinkinBorodinTau2007}
interpreted discrete isomonodromy through modifications of rational
$d$-connections and related gap probabilities to their $\tau$-functions.
The $q$-difference and symmetric settings were developed further in
\cite{knizel2015,OR2016,DzhamayKnizel2019}.
Polynomial weight changes and their normalization factors are related
to averages of characteristic polynomials
\cite{BaikDeiftStrahov2003}; Darboux transformations of Riemann--Hilbert
problems also compare gap probabilities
in \cite{BertolaCafasso2015}.

We describe the Christoffel transformation, cutoff evolution and
normalization in terms of symmetric $d$-connections. The substitution
$y=x^2$ identifies the nodes $x,-x$ and turns the unit shift into the
correspondence $z^2\mapsto(z+1)^2$, reversed by $z\mapsto-z-1$.
The weight transformation is represented by elementary modifications
at $0$ and $\infty$. On the common stable loci, these modifications
conjugate the point-addition maps and determine the change in the
normalization factor entering the gap probability recurrence.

The matrix-divisor construction \cite{Sheinman2017,Tyur65} defines the
underlying bundle without assuming that the normalized discrete
Riemann--Hilbert problem is solvable. For finite support and complex
weights, the moment matrix $H_w$ represents the connecting homomorphism
of a twisted bundle extension (Lemma~\ref{lem:rank-two-moment-cohomology}).
The normalized problem is therefore solvable exactly when $\det H_w\ne0$,
and Corollary~\ref{cor:rank-two-corank-splitting} gives the splitting
\[
 \mathcal E\simeq\mathcal O(r)\oplus\mathcal O(-r),
 \qquad r=\operatorname{corank}H_w.
\]
The elementary modifications for $w$ and $x^2w$ compare normalizations
through a signed ratio of moment determinants. For positive weights,
the determinant ratio equals the conditional expectation of the product
of squared particle coordinates (Lemma~\ref{lem:christoffel-cohomology}
and Section~\ref{sec:christoffel-gap}).

For positive even weights satisfying the moment and support assumptions,
we compare gap probabilities through moment determinants without imposing
a Pearson condition. For weights on the unit lattice, the bundle
construction gives a local recurrence formula for rational symmetric
Pearson ratios of arbitrary degree. For source ratios
$\varrho(z)=d(-z-1)/d(z)$ with $\deg d=1,2$, we express the Christoffel
transformation and the resulting recurrence in rational spectral
coordinates under the stated geometric hypotheses. The local formula
remains applicable outside these charts under its own hypotheses.

Section~\ref{sec:discrete-rhp} establishes the normalization and splitting
criteria. Section~\ref{sec:matrix-divisor-dynamics} describes point addition
and its Lax pair. Section~\ref{sec:quadratic-christoffel} constructs the
Christoffel modifications and the spectral charts, and
Section~\ref{sec:christoffel-gap} derives the probability comparison and
the recurrence coefficients.

\section{Discrete Riemann--Hilbert problem}
\label{sec:discrete-rhp}
\begin{dfn}[{\cite[Section~2.1]{Borodin_2003}}]
Let $\mathfrak X\subset\mathbb C$ be locally finite and
$\omega:\mathfrak X\to\operatorname{Mat}(2,\mathbb C)$.
A solution of the \emph{discrete Riemann--Hilbert problem} (DRHP)
$(\mathfrak X,\omega)$ is a matrix-valued function $m$, holomorphic off
$\mathfrak X$, with at most simple poles and satisfying
\begin{equation}\label{DRHP}
 \operatorname{Res}_{z=x}m(z)=\lim_{z\to x}m(z)\omega(x),
 \qquad x\in\mathfrak X.
\end{equation}
\end{dfn}
Equivalently, $m(z)(I-\omega(x)/(z-x))$ is holomorphic at $x$.
For $\omega(x)^2=0$, this gives the local factorization
\begin{equation}\label{cond}
 m(z)=H_x(z)\left(I+\frac{\omega(x)}{z-x}\right),
 \qquad H_x\text{ holomorphic at }x.
\end{equation}
In rank two, a pole implies $\omega(x)^2=0$: writing
$m=A_{-1}/(z-x)+A_0+O(z-x)$, the residue condition gives
$A_{-1}\omega=0$ and $A_{-1}=A_0\omega$. Hence $\det\omega=0$, and
the Cayley--Hamilton theorem \cite[Theorem~2.4.3.2]{HornJohnson2013}
gives $0=A_0\omega^2=(\operatorname{tr}\omega)A_{-1}$,
so $\omega^2=0$.

\begin{cor}[{\cite[Lemma~2.1]{Borodin_2003}}]\label{det_cor}
For a DRHP solution $m$, if every $\omega(x)$ is nilpotent,
then $\det m$ is entire.
If also $\det m\to1$ at infinity, then $\det m=1$.
\end{cor}
Henceforth the residue matrix and the associated bilinear form are
\[
 \omega(x)=\begin{pmatrix}0&w(x)\\0&0\end{pmatrix},\qquad
 (f,g)_w=\sum_{x\in\mathfrak X}f(x)g(x)w(x),
 \qquad w:\mathfrak X\to\mathbb C.
\]

\begin{thm}[{\cite[Theorem~2.4 and Remark~2.5]{Borodin_2003}}]\label{thm2}
Let $\mathfrak X\subset\mathbb C$ be nonempty and locally finite. Set $M=|\mathfrak X|-1$
(with $M=\infty$ for infinite support) and assume
\[
 \sum_{x\in\mathfrak X}|w(x)|\,|x|^j<\infty
 \qquad (j=0,1,2,\ldots).
\]
Suppose $(\cdot,\cdot)_w$ is nondegenerate on
$\mathbb C[z]_{\le n}$ for every $0\le n\le M$.
Let $P_n$ be the monic orthogonal polynomials and set
\[
 h_n=(P_n,P_n)_w\ne0,\qquad
 (\mathcal C_w p)(z)=\sum_{x\in\mathfrak X}\frac{p(x)w(x)}{z-x}.
\]
For infinite $\mathfrak X$, assume that nested bounded Jordan domains
exhaust $\mathbb C$ with boundaries uniformly separated from
$\mathfrak X$, and interpret asymptotics uniformly on these boundaries.
For $1\le k\le M$, the unique DRHP solution normalized by
\begin{equation}\label{asymp_DRHP}
 m_{\mathfrak X}(z)\operatorname{diag}(z^{-k},z^k)
 =I+O(z^{-1}),\qquad z\to\infty,
\end{equation}
is
\[
 m_{\mathfrak X}(z)=
 \begin{pmatrix}
 P_k(z)&(\mathcal C_wP_k)(z)\\
 h_{k-1}^{-1}P_{k-1}(z)&h_{k-1}^{-1}(\mathcal C_wP_{k-1})(z)
 \end{pmatrix}.
\]
\end{thm}
We encode the residue conditions by matrix divisors, following
\cite{Sheinman2017,Tyur65}. Let $\mathcal L$ be a rank-$r$ vector bundle on a
compact connected Riemann surface $S$. Write $\mathcal O(\mathcal L)$
for its sheaf of holomorphic sections, $\mathcal L_x$ for its fibre,
and $\mathcal O_x=\mathcal O_{S,x}$.
The twist $\mathcal L(D)=\mathcal L\otimes\mathcal O_S(D)$ allows
poles bounded by an effective divisor $D$.

\begin{dfn}\label{D:sect}
In fixed local holomorphic frames of $\mathcal L$, a \emph{matrix divisor}
is a collection of invertible meromorphic matrix germs $\Psi_x$, regular
and invertible outside a finite set $\mathfrak X$, modulo
$\Psi_x\mapsto C_x\Psi_x$ with $C_x$ regular and invertible.
Its sheaf of sections $\Gamma_{\mathbb C^r}(\Psi)$ consists of the
meromorphic sections $f$ of $\mathcal L$ satisfying
\[
 \Psi_x f_x\in\mathcal O(\mathcal L)_x\qquad(\forall x\in U)
\]
on each open set $U\subset S$.
\end{dfn}
This condition is independent of the representative.
Applying $\Psi_x^{-1}$ to a local frame gives a frame of the new
sheaf. Thus it defines a rank-$r$ bundle $\mathcal E$, equal to
$\mathcal L$ outside $\mathfrak X$.

\begin{lemma}
If $\det\Psi_x=1$ for every $x\in\mathfrak X$, then
\[
 \det\mathcal E\simeq\det\mathcal L,\qquad
 \chi(\mathcal E)=\deg\mathcal L+r(1-g),
\]
where $g$ is the genus of $S$. On $\mathbb P^1$, if $\deg\mathcal L=0$,
this gives $h^0(\mathcal E)=r+h^1(\mathcal E)$.
\end{lemma}
\begin{proof}
The local frames related by $\Psi_x^{-1}$ have the same top exterior
product. The Euler characteristic follows from the Riemann--Roch theorem
for vector bundles \cite[Chapter~VI, Theorem~4.4]{Freitag2013}.
\end{proof}

We now take $S=\mathbb P^1$, finite $\mathfrak X\subset\mathbb C$,
arbitrary complex weights, and $k\ge1$. No nondegeneracy is assumed.
Using $z$ for the affine coordinate, take
\[
 \mathcal L=\mathcal O(k)\oplus\mathcal O(-k),\qquad
 g_{0\infty}=\operatorname{diag}(z^{-k},z^k),\qquad
 f_\infty=g_{0\infty}f_0.
\]
The matrix divisor is represented by
\[
 \Psi_x=I-\frac{\omega(x)^T}{z-x},\qquad
 \Psi_x^{-1}=I+\frac{\omega(x)^T}{z-x}\quad(x\in\mathfrak X).
\]
Set $\mathcal E=\Gamma_{\mathbb C^2}(\Psi)$ and write
$H^i(\mathcal F)=H^i(\mathbb P^1,\mathcal F)$.
By \eqref{cond}, the columns of $m^T$ satisfy the local conditions
of $\mathcal E$ whenever $m$ solves the DRHP.
The local condition at $x\in\mathfrak X$ is
\begin{equation}\label{eq:rank-two-local-condition}
 \binom pq\in\mathcal O(\mathcal E)_x
 \quad\Longleftrightarrow\quad
 p\in\mathcal O_x,\qquad
 q-\frac{w(x)p}{z-x}\in\mathcal O_x.
\end{equation}
The inclusion $q\mapsto(0,q)^T$ and projection $(p,q)^T\mapsto p$ give
\begin{equation}\label{eq:rank-two-extension}
 0\longrightarrow\mathcal O(-k)
 \longrightarrow\mathcal E
 \longrightarrow\mathcal O(k)\longrightarrow0.
\end{equation}
Indeed, a germ $p$ lifts to $(p,w(x)p/(z-x))^T$, and the kernel is
$\mathcal O(-k)$. Away from $\mathfrak X$ this sequence is split.

Evaluation in the fixed frame at infinity gives
\begin{equation}\label{eq:rank-two-evaluation}
 0\longrightarrow\mathcal E(-\infty)
 \longrightarrow\mathcal E
 \xrightarrow{\operatorname{ev}_\infty}\mathbb C_\infty^2
 \longrightarrow0,
\end{equation}
where $\operatorname{ev}_\infty(f)=(g_{0\infty}f)(\infty)$ and
$\mathbb C_\infty^2$ is the skyscraper sheaf of the fibre at infinity.
Twisting \eqref{eq:rank-two-extension} yields
\begin{equation}\label{eq:rank-two-twisted-extension}
 0\longrightarrow\mathcal O(-k-1)
 \longrightarrow\mathcal E(-\infty)
 \longrightarrow\mathcal O(k-1)\longrightarrow0.
\end{equation}
The twist leaves the affine residue conditions unchanged.

Define the moments and the Hankel moment matrix by
\[
 \mu_n=\sum_{x\in\mathfrak X}w(x)x^n\quad(n\ge0),\qquad
 H_w=(\mu_{i+j})_{0\le i,j<k},\qquad
 \Delta_k=\det H_w.
\]

\begin{lemma}[Normalization criterion]\label{lem:rank-two-moment-cohomology}
The connecting homomorphism of
\eqref{eq:rank-two-twisted-extension} is represented by
$H_w$ in the monomial basis and its Serre-dual basis. Thus
\[
 H^0(\mathcal E(-\infty))\simeq\ker H_w,
 \qquad H^1(\mathcal E(-\infty))\simeq\operatorname{coker}H_w.
\]
The DRHP on $\mathfrak X$ has a solution with normalization
\eqref{asymp_DRHP} if and only if $\Delta_k\ne0$.
In this case evaluation at infinity is an isomorphism
$H^0(\mathcal E)\simeq\mathbb C^2$, and the unique lifts of its
standard basis form a global frame $F=m_{\mathfrak X}^T$ with
$\det F=1$. In particular, the normalized solution is unique,
$\mathcal E\simeq\mathcal O^{\oplus2}$, $h^0(\mathcal E)=2$,
and $H^1(\mathcal E)=0$.
\end{lemma}

\begin{proof}
The cohomology of line bundles on $\mathbb P^1$ and the Serre duality
theorem \cite[Chapter~VI, Theorem~6.1]{Freitag2013} give
\begin{equation}\label{eq:rank-two-twisted-cohomology}
 0\longrightarrow H^0(\mathcal E(-\infty))
 \longrightarrow\mathbb C[z]_{\le k-1}
 \xrightarrow{\delta_w}(\mathbb C[z]_{\le k-1})^*
 \longrightarrow H^1(\mathcal E(-\infty))\longrightarrow0.
\end{equation}
Lift $p$ to $(p,0)^T$ away from $\mathfrak X$ and to
$(p,w(x)p/(z-x))^T$ near $x$. Use the latter lift minus
the former for the boundary map. The residue pairing gives
\[
 \langle\delta_w(p),q\rangle
 =\sum_{x\in\mathfrak X}\operatorname{Res}_{z=x}
     \frac{w(x)p(x)q(z)\,dz}{z-x}
 =\sum_{x\in\mathfrak X}w(x)p(x)q(x),
 \qquad p,q\in\mathbb C[z]_{\le k-1}.
\]
Substituting monomials gives $H_w$; exactness gives the stated
identifications.

If $\Delta_k\ne0$, the twist is acyclic, so
\eqref{eq:rank-two-evaluation} makes evaluation an isomorphism.
The wedge of the normalized lifts is a section of
$\det\mathcal E\simeq\mathcal O$ equal to $1$ at infinity.
Thus they form a unique frame of determinant one, whose transpose
solves the normalized DRHP. Conversely, a normalized solution has
determinant one by Corollary~\ref{det_cor}. Its transpose trivializes
$\mathcal E$, so the twist is acyclic and $\Delta_k\ne0$.
\end{proof}

\begin{cor}[Moment corank and splitting]\label{cor:rank-two-corank-splitting}
For $r=\operatorname{corank}H_w$, even when $\Delta_k=0$,
\[
 \mathcal E\simeq\mathcal O(r)\oplus\mathcal O(-r),\qquad
 \mathcal E(-\infty)\simeq\mathcal O(r-1)\oplus\mathcal O(-r-1).
\]
\end{cor}
\begin{proof}
Since $\det\mathcal E\simeq\mathcal O$, the Birkhoff--Grothendieck theorem
\cite[Chapter~I, Theorem~2.1.1]{OSS2011} gives
$\mathcal E\simeq\mathcal O(a)\oplus\mathcal O(-a)$, $a\ge0$.
The preceding lemma gives $a=h^0(\mathcal E(-\infty))=r$.
\end{proof}

\begin{remark}
The groups $H^0(\mathcal E(-\infty))$ and $H^1(\mathcal E(-\infty))$
describe the ambiguity and obstruction to prescribing the value at infinity.
Vanishing of $H^1(\mathcal E)$ alone does not imply normalization:
$\mathcal O(1)\oplus\mathcal O(-1)$ is a counterexample.
\end{remark}

We use the following normalization principle for both point addition
and the Christoffel transformation.
For a vector bundle $B$ on $\mathbb P^1$, its \emph{determinant of cohomology}
is the one-dimensional space
\[
 \det R\Gamma(B):=\det H^0(B)\otimes(\det H^1(B))^{-1}.
\]
Here $\det V$ is the top exterior power of $V$, and the inverse of a line
denotes its dual. For acyclic $B$, this line is canonically $\mathbb C$.
For rank two and degree $-2$, we use
$\tau(B)\in(\det R\Gamma(B))^{-1}$ from
\cite[Definitions~1.9--1.10]{ArinkinBorodinTau2007}, equal to $1$ for
acyclic $B$ and $0$ otherwise.

For a finite-length torsion sheaf $Q$ on $\mathbb P^1_{\mathbb C}$,
$\operatorname{length}Q=\dim_{\mathbb C}H^0(Q)$.
In both applications below, the quotient sheaves have length one.

\begin{lemma}[Normalization via a common upper modification,
{\cite[Example~1.12]{ArinkinBorodinTau2007}}]
\label{lem:common-quotient-normalization}
Let $B,B',U$ be rank-two vector bundles on $\mathbb P^1$ with exact sequences
\[
 0\to B\to U\xrightarrow{\rho}Q\to0,\qquad
 0\to B'\to U\xrightarrow{\ell}Q'\to0,
\]
where $Q,Q'$ are torsion sheaves of equal finite length and $B$ is acyclic.
Put $V=H^0(Q)$, $V'=H^0(Q')$. Then $H^0(\rho)$ is an isomorphism,
$H^1(U)=0$, and, for $T=H^0(\ell)\circ H^0(\rho)^{-1}$,
\[
 0\to H^0(B')\to V\xrightarrow{T}V'\to H^1(B')\to0.
\]
Thus $B'$ is acyclic exactly when $T$ is invertible. Under the identification
\[
 \det R\Gamma(B)\otimes\det R\Gamma(B')^{-1}
 \simeq\det V'\otimes(\det V)^{-1}
\]
induced by the two sequences, $\tau(B')/\tau(B)=\det T$.
This is a scalar identity after bases of $V,V'$ are fixed.
\end{lemma}
\begin{proof}
The first cohomology sequence identifies $H^0(U)$ with $V$ and gives
$H^1(U)=0$; the second gives the displayed exact sequence.
Riemann--Roch and equality of the quotient lengths give
$\deg B=\deg B'=-2$.
The complex $[V\xrightarrow{T}V']$, in degrees zero and one, computes
$R\Gamma(B')$ using the acyclic trivialization for $B$.
Its determinant line is $\det V\otimes(\det V')^{-1}$, and
$\det T$ is its canonical section in the dual line, equal to the
acyclic unit when $T$ is invertible and zero otherwise.
\end{proof}

\section{Discrete dynamics of matrix divisors and a Lax pair}
\label{sec:matrix-divisor-dynamics}

Fix $k\ge1$, distinct points $\pi_0,\pi_1,\ldots\in\mathbb C$
and nonzero weights $w(\pi_j)$. For the finite sets
$\mathfrak X_s=\{\pi_0,\ldots,\pi_{s-1}\}$, let $\mathcal E_s$
be the bundle of Section~\ref{sec:discrete-rhp} and put
\[
 \widetilde{\mathcal E}_s=\mathcal E_s(-\infty),\qquad
 \widetilde{\mathcal L}=\mathcal O(k-1)\oplus\mathcal O(-k-1).
\]
In the affine trivialization by $1$, twisting leaves the residue
conditions unchanged and changes the transition matrix to
\begin{equation}\label{eq:dynamics-twisted-transition}
 \widetilde g_{0\infty}(z)=z g_{0\infty}(z)
 =\operatorname{diag}(z^{1-k},z^{k+1}).
\end{equation}
We use only indices for which these points exist. Write
\[
 \mu_n(s)=\sum_{x\in\mathfrak X_s}w(x)x^n,\qquad
 H_s=(\mu_{i+j}(s))_{0\le i,j<k},\qquad
 \Delta_k(s)=\det H_s.
\]
By Lemma~\ref{lem:rank-two-moment-cohomology},
$\Delta_k(s)\ne0$ is equivalent to acyclicity of $\widetilde{\mathcal E}_s$.
In this case the normalized DRHP matrix gives the frame
\begin{equation}\label{eq:dynamics-rational-frame}
 F_s=m_s^T,\qquad m_s:=m_{\mathfrak X_s},\qquad
 g_{0\infty}F_s=I+O(z^{-1}),\qquad \det F_s=1.
\end{equation}
It trivializes $\mathcal E_s$ and identifies
$\mathcal O(-\infty)^{\oplus2}$ with $\widetilde{\mathcal E}_s$.

For $\Delta_k(s)\ne0$, put $e(z)=(1,z,\ldots,z^{k-1})^T$.
The reproducing kernel on $\mathbb C[z]_{\le k-1}$ is characterized by
\begin{equation}\label{eq:dynamics-reproducing-kernel}
 \begin{gathered}
 K_{k-1}^{(s)}(z,u)=e(z)^T H_s^{-1}e(u),\\
 \sum_{x\in\mathfrak X_s}w(x)p(x)K_{k-1}^{(s)}(x,u)=p(u),
 \qquad p\in\mathbb C[z]_{\le k-1}.
 \end{gathered}
\end{equation}
At $s=k$, Vandermonde factorization gives $\Delta_k(k)\ne0$.

\begin{lemma}[Adding a point]\label{lem:dynamics-point-addition}
Assume $\Delta_k(s)\ne0$. Let
$\mathcal U_s\subset\widetilde{\mathcal E}_s(\pi_s)$ allow a
simple pole at $\pi_s$ with residue in $\mathbb C e_2$,
\begin{equation}\label{eq:dynamics-two-sequences}
 \begin{aligned}
 0&\longrightarrow\widetilde{\mathcal E}_s
 \longrightarrow\mathcal U_s
 \xrightarrow{\rho_s}\mathbb C_{\pi_s}\longrightarrow0,\\
 0&\longrightarrow\widetilde{\mathcal E}_{s+1}
 \longrightarrow\mathcal U_s
 \xrightarrow{\ell_s}\mathbb C_{\pi_s}\longrightarrow0,
 \end{aligned}
\end{equation}
where $\rho_s(p,q)=\operatorname{Res}_{z=\pi_s}q(z)$ and
$\ell_s(p,q)=\rho_s(p,q)-w(\pi_s)p(\pi_s)$.
These formulas also define the quotient maps after twisting by
$\mathcal O(\infty)$.
The unique section $h_s\in H^0(\mathcal U_s)$ with $\rho_s(h_s)=1$ is
\begin{equation}\label{eq:dynamics-normalizing-section}
 h_s(z)=\begin{pmatrix}
 -K_{k-1}^{(s)}(z,\pi_s)\\[4pt]
 \displaystyle\frac1{z-\pi_s}
 -\sum_{x\in\mathfrak X_s}
 \frac{w(x)K_{k-1}^{(s)}(x,\pi_s)}{z-x}
 \end{pmatrix}.
\end{equation}
With
\begin{equation}\label{eq:dynamics-epsilon-kernel}
 \varepsilon_s=\ell_s(h_s)
 =1+w(\pi_s)K_{k-1}^{(s)}(\pi_s,\pi_s)
 =\frac{\Delta_k(s+1)}{\Delta_k(s)},
\end{equation}
the cohomology of the new bundle is given by
\begin{equation}\label{eq:dynamics-scalar-cohomology}
 0\longrightarrow H^0(\widetilde{\mathcal E}_{s+1})
 \longrightarrow\mathbb C\xrightarrow{\varepsilon_s}\mathbb C
 \longrightarrow H^1(\widetilde{\mathcal E}_{s+1})
 \longrightarrow0.
\end{equation}
If $\varepsilon_s\ne0$, then
\begin{equation}\label{eq:dynamics-first-lax}
 \begin{gathered}
 m_{s+1}(z)=\left(I+\frac{N_s}{z-\pi_s}\right)m_s(z),\qquad N_s^2=0,\\
 N_s=\frac{m_s(\pi_s)\omega(\pi_s)m_s(\pi_s)^{-1}}{\varepsilon_s}.
 \end{gathered}
\end{equation}
\end{lemma}
\begin{proof}
The two kernels impose the absence of a pole and the new residue
condition, respectively. Lemma~\ref{lem:common-quotient-normalization}
identifies $H^0(\mathcal U_s)$ with $\mathbb C$ via $\rho_s$ and gives
\eqref{eq:dynamics-scalar-cohomology}, with $T=\ell_s(h_s)$.
The displayed $h_s$ has the required residues. Its second entry
is $O(z^{-k-1})$ by \eqref{eq:dynamics-reproducing-kernel}, so it is
holomorphic at infinity in the twisted frame.
Applying $\ell_s$ and the rank-one determinant identity gives
\eqref{eq:dynamics-epsilon-kernel}.

For $\varepsilon_s\ne0$, the normalized transfer is
\begin{equation}\label{eq:dynamics-section-projection}
 H^0(\mathcal E_s)\longrightarrow H^0(\mathcal E_{s+1}),\qquad
 f\longmapsto f-\varepsilon_s^{-1}\ell_s(f)h_s.
\end{equation}
Its image lies in the kernel of $\ell_s$, and it preserves evaluation at infinity since
$g_{0\infty}h_s=O(z^{-1})$. Uniqueness gives
\[
 h_s(z)=\frac{F_s(z)F_s(\pi_s)^{-1}e_2}{z-\pi_s}.
\]
The right-hand side has residue $e_2$ and is regular at infinity
in the twisted frame. Applying the transfer to the columns of $F_s$
and transposing proves \eqref{eq:dynamics-first-lax}, extending
\cite[Theorem~3.1(a), Lemma~5.1]{Borodin_2003} to complex weights.
\end{proof}

\begin{remark}
The sequences \eqref{eq:dynamics-two-sequences} remain exact even
when a moment determinant vanishes. Their equal length-one quotients
preserve $\det\widetilde{\mathcal E}_s\simeq\mathcal O(-2\infty)$.
If $\Delta_k(s)\ne0$ and $\varepsilon_s=0$, then
$h^0(\widetilde{\mathcal E}_{s+1})=h^1(\widetilde{\mathcal E}_{s+1})=1$.
On the nondegenerate locus, the rational identity
$\iota_s:\widetilde{\mathcal E}_s\dashrightarrow
\widetilde{\mathcal E}_{s+1}$ has matrix
\begin{equation}\label{eq:dynamics-identity-matrix}
 B_s(z)=F_{s+1}(z)^{-1}F_s(z)=I-\frac{N_s^T}{z-\pi_s}.
\end{equation}
Unlike \eqref{eq:dynamics-section-projection}, this identity need not
preserve global holomorphic sections after twisting by
$\mathcal O(\infty)$.
\end{remark}

\subsection{Transport along the lattice}

Assume in addition that
\begin{equation}\label{eq:dynamics-pearson}
 \begin{gathered}
 \sigma(z)=\eta z+\beta,\qquad \eta\ne0,\qquad
 \sigma(\pi_{j+1})=\pi_j,\\
 D(z)=\operatorname{diag}(d_1(z),d_2(z)),\qquad d_1(\pi_0)=0,\\
 \frac{w(\pi_{j-1})}{w(\pi_j)}
 =\eta\frac{d_1(\pi_j)}{d_2(\pi_j)}\qquad(j\ge1).
 \end{gathered}
\end{equation}
Here $d_1,d_2$ are nonzero polynomials satisfying
$d_2(\pi_j)\ne0$ at all relevant indices $j\ge1$.
Put $n=\max(\deg d_1,\deg d_2)$.
Identify $\sigma^*\mathcal O(-\infty)$ with
$\mathcal O(-\infty)$ by $1$ in the affine trivializations.
Multiplication by $D$ then defines
\[
 \mathscr D_s:\sigma^*\widetilde{\mathcal E}_s
 \longrightarrow\widetilde{\mathcal E}_{s+1}(n\infty).
\]

\begin{lemma}[Lax pair, \cite{Borodin_2003}]\label{lem:dynamics-lax-pair}
If $\Delta_k(s)\Delta_k(s+1)\ne0$, the matrix of $\mathscr D_s$
in the frames $F_s(\sigma z),F_{s+1}(z)$ is $M_s(z)^T$, where
\begin{equation}\label{eq:dynamics-second-lax}
 m_s(\sigma z)D(z)=M_s(z)m_{s+1}(z).
\end{equation}
It is polynomial, with $\deg M_s\le n$ and $\det M_s=d_1d_2$.
If $d_i(z)=\lambda_i z^n+O(z^{n-1})$, then
\[
 M_s(z)=\operatorname{diag}(\eta^k\lambda_1,\eta^{-k}\lambda_2)
 z^n+O(z^{n-1}).
\]
If also $\Delta_k(s+2)\ne0$, then
\begin{equation}\label{eq:dynamics-compatibility}
 B_{s+1}(z)M_s(z)^T=M_{s+1}(z)^T B_s(\sigma z).
\end{equation}
\end{lemma}
The proofs of \cite[Theorems~3.1(b,c) and~5.2(a)]{Borodin_2003}
use positivity only to ensure existence of the normalized DRHP solutions.
By Lemma~\ref{lem:rank-two-moment-cohomology}, the stated determinant
conditions suffice, so the same proofs apply to complex weights. For the last available node, see
\cite[Remark~3.2]{Borodin_2003}.

The compatibility equation gives a rational update.
Put $a=\pi_{s+1}$, so $\det M_s(a)=d_1(a)d_2(a)\ne0$.
Where the scalar denominator is nonzero,
\begin{equation}\label{eq:dynamics-rational-update}
 \begin{aligned}
 N_{s+1}
 &=\frac{M_s(a)^{-1}N_sM_s(a)}
 {\eta+\operatorname{tr}(N_sM_s'(a)M_s(a)^{-1})},\\
 M_{s+1}(z)
 &=\left(I+\frac{N_s}{\eta(z-a)}\right)M_s(z)
 \left(I-\frac{N_{s+1}}{z-a}\right).
 \end{aligned}
\end{equation}
Indeed, the vanishing of the simple pole is equivalent to
\[
 \bigl(M_s(a)+\eta^{-1}N_sM_s'(a)\bigr)N_{s+1}
 =\eta^{-1}N_sM_s(a).
\]
The rank-one identity $N_sBN_s=\operatorname{tr}(N_sB)N_s$ solves this
equation with the displayed denominator and gives
$N_sM_s(a)N_{s+1}=0$. This cancels the double pole, so $M_{s+1}$ is polynomial.

\subsection{Relation to rational \texorpdfstring{$d$}{d}-connections}

Assume $\sigma(z)=z-1$.
Here a rational $d$-connection is a generically invertible rational
transport from the fibre at $z$ to that at $z+1$
\cite[Definition~2.1]{arinkin2007}.
The fixed ambient transport
\[
 \frac{D(z+1)}{d_1(z+1)}
 =\operatorname{diag}\left(1,\frac{d_2(z+1)}{d_1(z+1)}\right)
\]
induces a $d$-connection on each $\widetilde{\mathcal E}_s$.
When $\Delta_k(s)\ne0$, its matrix in the rational frame $F_s$ is
\[
 A_s(z)=\frac{F_s(z+1)^{-1}D(z+1)F_s(z)}{d_1(z+1)}.
\]
If also $\Delta_k(s+1)\ne0$, the second Lax equation yields
\[
 A_s(z)=\frac{B_s(z+1)^{-1}M_s(z+1)^T}{d_1(z+1)},
\]
and the modification induced by the rational identity satisfies
\begin{equation}\label{eq:dynamics-forward-connection}
 A_{s+1}(z)=B_s(z+1)A_s(z)B_s(z)^{-1}.
\end{equation}
Thus point addition changes the bundle while preserving the rational
transport. The diagonal transport preserves
\[
 \mathcal O(-k-1)\hookrightarrow\widetilde{\mathcal E}_s,
 \qquad q\longmapsto(0,q)^T,
\]
so these connections are reducible. An upper elementary modification
changes degree $-2$ to $-1$ and preserves the invariant rational line.

\section{Moduli spaces under the Christoffel transformation}
\label{sec:quadratic-christoffel}

We define the moduli problems $\mathcal M_s$ and $\widehat{\mathcal M}_s$
by prescribed local data and data at infinity and construct their point-addition
maps from local lattices. Elementary modifications identify the common
stable loci. Pearson data of degree one or two provide rational spectral
coordinates for both problems.
For Pearson weights, we assume $\epsilon=1$ in
$\varrho(z)=\epsilon d(-z-1)/d(z)$, as for the symmetric Hahn and
Krawtchouk weights \cite[Equations~(4.1),~(4.7)]{NNS2026}.

\subsection{Weights, bundles and normalization}
\label{subsec:symmetric-reduction}

Fix $k\ge2$ and a locally finite symmetric set $\mathfrak X\subset\mathbb R$.
Let $w:\mathfrak X\to\mathbb R_{>0}$ be even, with
\begin{equation}\label{eq:positive-moment-assumption}
 \sum_{x\in\mathfrak X}(1+|x|^{4k-2})w(x)<\infty.
\end{equation}
Put $\mathfrak X_s=\mathfrak X\cap[-s,s]$, $\widehat w=x^2w$ and
$y=z^2$. At every cutoff used, assume at least $k$ distinct positive
squared nodes. For $v=w,\widehat w$, write
\[
 \mu_s[v]=\sum_{x\in\mathfrak X_s}v(x)\delta_{x^2},\qquad
 \Delta_m[v,s]=\det\left(\sum_{x\in\mathfrak X_s}x^{2(i+j)}v(x)\right)_{0\le i,j<m}.
\]
Here $\delta_a$ denotes the unit point mass at $a$.
Thus each pair of nodes $x,-x\in\mathfrak X_s$, with $x>0$,
contributes mass $2v(x)$ at $x^2$; if $0\in\mathfrak X_s$,
the mass at zero is $v(0)$.
If $0\notin\mathfrak X$, we set $w(0)=\widehat w(0)=0$.

Set $R_s[v](y)=\sum_{x\in\mathfrak X_s}v(x)/(y-x^2)$.
The basic identity for the weight transformation is
\begin{equation}\label{eq:christoffel-cauchy-identity}
 R_s[\widehat w]=yR_s[w]-\sum_{x\in\mathfrak X_s}w(x).
\end{equation}
A holomorphic local section of $\mathcal G_s[v]$ is a meromorphic pair
$(\phi,\chi)^T$ such that $\phi,\chi-R_s[v]\phi$ are holomorphic at finite
points and $\phi=O(y^{k-1})$, $\chi=O(y^{-k})$ at infinity.
Its elementary lower modification at infinity is defined by
\begin{equation}\label{eq:symmetric-polynomial-kernel}
 0\longrightarrow\mathcal F_s[v]\longrightarrow\mathcal G_s[v]
 \xrightarrow{\ (\phi,\chi)^T\mapsto[y^{-k}]\chi\ }\mathbb C_\infty
 \longrightarrow0.
\end{equation}
Here $\mathbb C_p$ is the length-one skyscraper sheaf at $p$, and brackets
extract Laurent coefficients. Thus $\mathcal F_s[v]$ imposes
$\chi=O(y^{-k-1})$, with degrees $\deg\mathcal G_s[v]=-1$ and
$\deg\mathcal F_s[v]=-2$.
The lifts
\[
 (\phi,\chi)^T\mapsto(\phi(z^2),z\chi(z^2))^T,\qquad
 (\phi,\chi)^T\mapsto(z\phi(z^2),\chi(z^2))^T
\]
identify the local sections of $\mathcal F_s[w]$ with sections having
even first and odd second components at index $2k$, and those of
$\mathcal F_s[\widehat w]$ with sections having odd first and even second
components at index $2k+1$.
The residue conditions and conditions at infinity in Sections~2--3 give
these identifications; their polynomial form is \cite[Equation~(3.8)]{NNS2026}.

In the local frames
$C_s[v]=\bigl(\begin{smallmatrix}1&0\\R_s[v]&1\end{smallmatrix}\bigr)$,
$H=\operatorname{diag}(1,y)$ satisfies
\begin{equation}\label{eq:alt-christoffel-lattice}
 H_s^{\mathrm{res}}:=C_s[\widehat w]^{-1}HC_s[w]
 =\begin{pmatrix}1&0\\\sum_{x\in\mathfrak X_s}w(x)&y\end{pmatrix}.
\end{equation}
Thus the weight transformation changes the bundle only at zero and
infinity. The following diagram gives their intrinsic description.

Let $\mathcal E$ have rank two and degree $-1$, with lines
$\ell_0\subset\mathcal E_0$ and $\ell_+\subset\mathcal E_\infty$.
Let $\mathcal U$ be its elementary upper modification at zero along
$\ell_0$ \cite[Example~1.6]{ArinkinBorodinTau2007}: its sections may have
a simple pole whose $y^{-1}$ coefficient lies in $\ell_0$.
Put $Q_0=\mathcal U/\mathcal E$, and let $Q_\infty$ be the skyscraper sheaf
with fibre $\mathcal E_\infty/\ell_+$. Define
\begin{equation}\label{eq:intrinsic-christoffel-sequences}
 \mathcal K=\ker(\mathcal E\longrightarrow Q_\infty),\qquad
 \widehat{\mathcal E}=\ker(\mathcal U\longrightarrow Q_\infty).
\end{equation}
The maps to $Q_\infty$ evaluate sections modulo $\ell_+$.
We view all four bundles inside the same meromorphic bundle.

\begin{lemma}[Commuting elementary modifications]
\label{lem:common-upper-modification}
The rows of the following diagram are exact, and both vertical
inclusions have cokernel $Q_0$:
\begin{equation}\label{eq:christoffel-diamond-sequences}
 \begin{tikzcd}[column sep=small,row sep=small]
 0\arrow[r]
 &\mathcal K\arrow[r,hook']\arrow[d,hook']
 &\mathcal E\arrow[d,hook']\arrow[dr]\\
 0\arrow[r]
 &\widehat{\mathcal E}\arrow[r,hook']
 &\mathcal U\arrow[r]
 &Q_\infty\arrow[r]&0.
 \end{tikzcd}
\end{equation}
Thus $\deg\mathcal K=-2$, $\deg\mathcal U=0$ and
$\deg\widehat{\mathcal E}=-1$, and
\begin{equation}\label{eq:christoffel-determinant-line}
 \det R\Gamma(\widehat{\mathcal E})\simeq
 \det R\Gamma(\mathcal E)\otimes H^0(Q_0)\otimes H^0(Q_\infty)^{-1}.
\end{equation}
\end{lemma}
\begin{proof}
Disjoint supports give the diagram and degrees. The long exact
cohomology sequences give the determinant-line identity.
\end{proof}

For the bundles associated with the weights, take $\ell_0$ to be the second axis in the
$C_s[w]$-frame at zero and $\ell_+$ the kernel of the coefficient
map in \eqref{eq:symmetric-polynomial-kernel} on the fibre at infinity.
Equation~\eqref{eq:alt-christoffel-lattice} identifies the bundles in the diagram as
\[
 (\mathcal E,\widehat{\mathcal E},\mathcal K)
 =(\mathcal G_s[w],H^{-1}\mathcal G_s[\widehat w],\mathcal F_s[w]).
\]
The two quotient maps on $\mathcal G_s[\widehat w]$ compare normalizations.

\begin{lemma}[Normalization under the Christoffel transformation]
\label{lem:christoffel-cohomology}
For $v=w,\widehat w$, the bundle $\mathcal F_s[v]$ is acyclic, and
\begin{equation}\label{eq:christoffel-mixed-sequences}
 \begin{aligned}
 0&\to\mathcal F_s[w]\xrightarrow H\mathcal G_s[\widehat w]
       \xrightarrow{\rho}\mathbb C_0\to0,\\
 0&\to\mathcal F_s[\widehat w]\to\mathcal G_s[\widehat w]
       \xrightarrow{\ell}\mathbb C_\infty\to0.
 \end{aligned}
\end{equation}
Here $\rho(\phi,\chi)=\chi(0)-w(0)\phi(0)$ and
$\ell(\phi,\chi)=[y^{-k}]\chi$.
If $\mathsf p_{k,s}$ is the monic degree-$k$ orthogonal polynomial for $\mu_s[w]$,
the unique global section $\mathbf h$ with $\rho(\mathbf h)=1$ satisfies
\[
 \varepsilon_H(s):=\ell(\mathbf h)=\mathsf p_{k,s}(0)=(-1)^k
 \frac{\Delta_k[\widehat w,s]}{\Delta_k[w,s]}\ne0.
\]
\end{lemma}

\begin{proof}
Lemma~\ref{lem:rank-two-moment-cohomology} gives acyclicity,
and the local conditions give both sequences.
Lemma~\ref{lem:common-quotient-normalization} identifies $H^0(\mathcal G_s[\widehat w])$
with $\mathbb C$ via $\rho$ and gives the $\tau$-ratio as $\ell(\mathbf h)$.
For $\mathbf h=(\phi,\chi)^T$, the local conditions also give
\[
 \begin{gathered}
 \sum_{x\in\mathfrak X_s}\phi(x^2)x^{2j}w(x)=-\delta_{j0}\quad(0\le j<k),\\
 \ell(\mathbf h)=\sum_{x\in\mathfrak X_s}\phi(x^2)x^{2k}w(x).
 \end{gathered}
\]
By orthogonality, subtracting $\mathsf p_{k,s}$ in the last sum gives
$\ell(\mathbf h)=\mathsf p_{k,s}(0)$. Cramer's rule gives the determinant ratio.
\end{proof}

Thus, in the quotient trivializations $\rho,\ell$,
\[
 \frac{\tau(\mathcal F_s[\widehat w])}{\tau(\mathcal F_s[w])}
 =\varepsilon_H(s).
\]

\subsection{The transformed moduli problem}
\label{subsec:symmetric-general-moduli}

For point addition, take $\mathfrak X=\{-N,\ldots,N\}$, $1<k<N-1$,
or $\mathfrak X=\mathbb Z$. The indices are $k\le s<N$ in the finite
case and $s\ge k$ otherwise. Put $\iota(z)=-z-1$.
For the unit shift, $\eta=1$ in \eqref{eq:dynamics-pearson}, and its
forward Pearson ratio is
\begin{equation}\label{eq:positive-pearson-assumptions}
 \begin{gathered}
 \frac{w(j+1)}{w(j)}=\varrho(j),\qquad
 \varrho(z)=\frac{d_2(z+1)}{\eta d_1(z+1)}
 =\frac{d(-z-1)}{d(z)},\\
 m=\deg d\in\{1,2\},\qquad \nu=[z^{m-1}]d(z).
 \end{gathered}
\end{equation}
Here $d$ is real and monic, $d$ and $d\circ\iota$ are coprime, and the ratio
is finite and positive on all lattice edges.
The ratio satisfies $\varrho(z)\varrho(-z-1)=1$; the sign
$\epsilon=1$ is part of our Pearson assumption.
For the moduli charts, the roots of $d(z)d(-z-1)$ and
the points $0,-1,s,-s-1$ must be pairwise distinct.
Constant ratios and coincident singularities are discussed below.

\begin{dfn}[Symmetric connection]
A symmetric connection on $\mathcal E$ over $\mathbb P^1_y$ is a rational
transport \cite[Definition~1.2.1]{DzhamayKnizel2019}
\[
 \nabla_z:\mathcal E_{z^2}\dashrightarrow\mathcal E_{(z+1)^2},\qquad
 \nabla_{\iota(z)}\nabla_z=\mathrm{id}.
\]
\end{dfn}
Its matrix satisfies $A\circ\iota=A^{-1}$. A change $\mathbf v'=R(y)\mathbf v$ gives
$A'(z)=R((z+1)^2)A(z)R(z^2)^{-1}$.
The transports induced by the parity decomposition in Section~3 are
\begin{equation}\label{eq:christoffel-even-connection}
 \begin{gathered}
 D_w(z)=\operatorname{diag}\left(1,\frac z{z+1}\varrho(z)\right),\quad
 D_{\widehat w}(z)=\operatorname{diag}\left(1,\frac{z+1}z\varrho(z)\right),\\
 D_{\widehat w}(z)H(z^2)=H((z+1)^2)D_w(z).
 \end{gathered}
\end{equation}
They act on $\mathcal G_s[w]$ and $\mathcal G_s[\widehat w]$, respectively.
For the weight $w$, the transport on sections with odd first component
and even second component differs from $D_{\widehat w}$ by the scalar
factor $z/(z+1)$.
\begin{dfn}[Stability]
\label{def:stable-symmetric-connection}
Let $(\mathcal E,\nabla)$ be a rank-two symmetric connection
with $\deg\mathcal E=-1$.
A line subbundle $L\subset\mathcal E$ with locally free quotient
is \emph{invariant} if
\[
 \nabla_z(L_{z^2})=L_{(z+1)^2}
 \qquad\text{for generic }z.
\]
The pair $(\mathcal E,\nabla)$ is \emph{stable} if every invariant
line subbundle satisfies $\deg L\le-1$, equivalently
$\deg L<\deg\mathcal E/2$.
\end{dfn}

\begin{remark}[Stability and irreducibility]
Such a connection is called \emph{irreducible} if it has
no invariant line subbundle, in analogy with rank-two
$d$-connections; see \cite[Lemma~3.12]{arinkin2007}.
Irreducibility implies stability, but not conversely.
For $v=w,\widehat w$, the second coordinate defines
$\mathcal O(-k)\subset\mathcal G_s[v]$ with quotient
$\mathcal O(k-1)$, preserved by the diagonal transports in
\eqref{eq:christoffel-even-connection}. Thus the weight
connections are reducible; their stability under the stated
hypotheses is verified below.
\end{remark}
The source problem $\mathcal M_s$ parametrizes stable rank-two,
degree-$-1$ symmetric connections with the following local data.
A connection has type $(r_1,r_2)$ at $a$, with $r_1\le r_2$, if its matrix is
$\operatorname{diag}((z-a)^{r_1},(z-a)^{r_2})$ in regular frames of the
pulled-back source and target bundles. Reflection gives type $(-r_2,-r_1)$ at $\iota(a)$.
We prescribe types $(0,1)$ at $0$, $(-1,0)$ at the roots of $d$, their reflected types,
and $(-1,1)$ at $s,-s-1$. All other finite points are regular.
Fix the determinant $z\varrho(z)/(z+1)$ in affine determinant frames.
Isomorphisms are regular invertible bundle maps intertwining transport;
modifications are generically invertible rational bundle maps
intertwining transport. We regard the moduli problems here as groupoids
of algebraic families, with the prescribed local types in relative
regular frames, the stated conditions at infinity, and stability on every
geometric fibre. Isomorphisms of moduli problems mean equivalences of
these groupoids compatible with base change. The spectral normal forms
below give varieties parametrizing isomorphism classes on the indicated
open charts.

The normalization in a holomorphic frame at infinity depends on $\deg d$.
For quadratic $d$, require
\begin{equation}\label{eq:christoffel-infinity-data}
 A=I+B/z+O(z^{-2}),\qquad
 \operatorname{Spec}B=(d_+,d_-),\qquad d_+\ne d_-,d_-+2.
\end{equation}
For linear $d$, require $A(z)=A_\infty+O(z^{-1})$ with
$\operatorname{Spec}A_\infty=\{1,-1\}$, and prescribe the eigenvalue
expansions $1+d_+/z+O(z^{-2})$ and $-(1+d_-/z+O(z^{-2}))$.
The coefficients are ordered, and $\gamma=(d_+-d_-)/2$.
Set $\ell_0=\ker\nabla_0$. The line $\ell_+$ is the
$d_+$ eigenline of $B$ in the quadratic case and the $+1$ eigenline
of $A_\infty$ in the linear case.
\begin{remark}[Compatibility at infinity]
The parameters at infinity of $\mathcal M_s$ satisfy
\[
 d_++d_-=m+1-2\nu.
\]
Indeed, $\deg\mathcal E=-1$ gives
$\det\mathcal E\simeq\mathcal O(-1)$. Passing from the affine
determinant frame to a compatible holomorphic determinant frame
at infinity multiplies the determinant of the transport by
$(z+1)^2/z^2$. Hence, in the latter frame, it equals
\[
 \frac{z+1}{z}\varrho(z)
 =(-1)^m\left(1+\frac{m+1-2\nu}{z}+O(z^{-2})\right).
\]
Here we used that $d$ is monic and $[z^{m-1}]d=\nu$.
The prescribed expansions at infinity give the determinant
$(-1)^m(1+(d_++d_-)/z+O(z^{-2}))$, proving the relation.
\end{remark}
\begin{thm}[Elementary modifications of symmetric connections]
\label{thm:alt-canonical-christoffel}
These lines determine elementary modifications defining a symmetric connection
$\mathcal H(\mathcal E,\nabla)$ with
\begin{equation}\label{eq:alt-christoffel-data}
 \begin{gathered}
 (0,1)_0\mapsto(-1,0)_0,\quad
 (-1,0)_{-1}\mapsto(0,1)_{-1},\quad
 (\widehat d_+,\widehat d_-)=(d_+,d_-+2),\\
 \det\widehat A(z)=\frac{(z+1)^2}{z^2}\det A(z),\qquad
 \det\widehat{\mathcal E}\simeq\det\mathcal E\otimes\mathcal O([0]-[\infty]).
 \end{gathered}
\end{equation}
Other local data are unchanged. Define $\widehat{\mathcal M}_s$ by
these data and stability, independently of the construction.
Let $\mathcal M_s^\circ,\widehat{\mathcal M}_s^\circ$ be the loci with
stable forward and inverse modifications, respectively. Then
\begin{equation}\label{eq:christoffel-moduli-isomorphism}
 \mathcal H_s:\mathcal M_s^\circ\xrightarrow{\sim}\widehat{\mathcal M}_s^\circ.
\end{equation}
The inverse consists of a lower modification at zero along the kernel
of the residue and an upper modification at infinity along the
eigenline labelled by $\widehat d_-$.
\end{thm}
\begin{proof}
At zero, a frame with second vector in $\ell_0$ gives
$A(z)=(a(z),zb(z))$, with $\det(a(0),b(0))\ne0$.
The upper modification divides its second column by $z^2$, producing
type $(-1,0)$. The kernel of the residue specifies the inverse operation.
Reflection gives the result at $-1$.
At infinity the new matrix is
\[
 \operatorname{diag}(1,(z+1)^2)A(z)\operatorname{diag}(1,z^{-2}).
\]
For distinct limiting eigenvalues, reflection makes the off-diagonal
entries $O(z^{-2})$ in an eigenframe. For $A=I+B/z+C/z^2+\cdots$ it gives
$C=(B^2-B)/2$, so they are $O(z^{-3})$ in a $B$-eigenframe.
In either case the new leading data are triangular, with the
eigenline labelled by $\widehat d_-$ as second axis and coefficients $(d_+,d_-+2)$.
The eigenvalues are distinct under the stated assumptions.
The upper modification along this line reverses the operation.
The determinant follows from the exact sequences.
The prescribed local types make the kernels at zero line subbundles
in families, and the distinct labelled eigenvalues give eigenline
subbundles at infinity by their spectral projectors.
Both constructions commute with base change, as do the elementary
modifications. The inverse constructions thus give an equivalence
of these moduli problems.
\end{proof}

For the source weight $w$, the flags are those in
\eqref{eq:alt-christoffel-lattice}. Their coefficients at infinity are
\begin{equation}\label{eq:weight-infinity-coefficients}
 d_+=2-2k,\qquad d_-=2k+m-1-2\nu,\qquad
 \gamma=\nu+\frac{3-m}{2}-2k.
\end{equation}
For $m=2$, the source and target together require $\gamma\ne0,1$;
positivity alone does not imply this. For $m=1$, no such exclusions are needed.
The normalization sequences give
$\mathcal G_s[v]\simeq\mathcal O\oplus\mathcal O(-1)$.
Its only possible destabilizing line is its unique degree-zero
summand. Both components of the generating section at infinity are nonzero:
they are the leading coefficient and positive squared norm of the
monic polynomial of degree $k-1$. The line it spans at infinity is not an eigenline
of the diagonal data at infinity in the weight frame. Invariance requires this.
For $m=2$, a local holomorphic generator of this subbundle at infinity has the form
$v_0+O(z^{-2})$, since $y=z^2$; comparing the $z^{-1}$ terms in the invariance
relation gives $Bv_0\in\langle v_0\rangle$. Thus the connections associated with
both weights are stable under the stated eigenvalue conditions.

To compare spectral coordinates, multiply the transformed connection by $z/(z+1)$:
\begin{equation}\label{eq:christoffel-original-moduli}
 \mathcal C(A)=\frac z{z+1}\mathcal H(A),\qquad
 (d_+,d_-)\mapsto(d_+-1,d_-+1),\qquad \widehat\gamma=\gamma-1.
\end{equation}
This restores the source determinant and finite types and preserves
stability. For $m=2$, the target exclusions are $\widehat\gamma\ne-1,0$.
The map $\mathcal H$ represents the weight transformation, and
$\mathcal C$ has the determinant normalization used below.

\subsection{Conjugacy of the point-addition maps}
\label{subsec:christoffel-dynamics-transfer}

Define the point-addition map $\Phi_s$ by the following modification,
for $k\le s\le N-2$ on the finite lattice and $s\ge k$ on $\mathbb Z$.
Assume the stated distinctness conditions at both indices $s,s+1$.
Write $\Lambda_b(\xi)$ for the pullback of the bundle lattice along
$y=(b+\xi)^2$. Change only the lattice at $(s+1)^2$, setting
\begin{equation}\label{eq:intrinsic-image-lattice}
 \Lambda'_{s+1}(\xi)=\nabla_{s+\xi}\Lambda_s(\xi),\qquad
 \Lambda'_{-s-1}(\xi)=\Lambda'_{s+1}(-\xi).
\end{equation}
The second equality identifies the two unramified branches over
$(s+1)^2$. This removes the singularities at $s,-s-1$ and gives type
$(-1,1)$ at $s+1,-s-2$; all other prescribed data are unchanged.
The degree is unchanged because this local type has total order zero.
The inverse changes the same lattice by
\[
 \Lambda_{s+1}(\xi)
 =(\nabla'_{s+1+\xi})^{-1}\Lambda'_{s+2}(\xi),
\]
with the reflected branch defined as above.
Indeed, before each change, the transport is regular at the adjacent
point, so the two image-lattice prescriptions are inverse.
We restrict both maps to the loci where the source and target are
stable. This is the symmetric-correspondence version of the
singularity shift in \cite[Lemma~1.7(3)--(4)]{ArinkinBorodinTau2007}.

\begin{thm}[Compatibility of elementary modifications]
\label{prop:intrinsic-pair-motion}
On the common stable loci at consecutive indices, the point-addition
map for the transformed weight is
\begin{equation}\label{eq:alt-christoffel-moduli-square}
 \widehat\Phi_s=\mathcal H_{s+1}\Phi_s\mathcal H_s^{-1}.
\end{equation}
The same conjugacy holds with $\mathcal C$ in place of $\mathcal H$.
If $B_s$ and $H_s^{\mathrm{Ch}}$ represent point addition and the
modification $\mathcal H_s$, respectively, then
\[
 \widehat B_s(y)=H_{s+1}^{\mathrm{Ch}}(y)B_s(y)H_s^{\mathrm{Ch}}(y)^{-1}.
\]
\end{thm}
\begin{proof}
The Christoffel modification changes lattices only at $0,\infty$; point
addition changes the lattice at $(s+1)^2$. Since $s\ge k>1$, point addition
also fixes the lattices defining the flags at $0,1,\infty$.
Thus both orders give the same local lattices and transport.
The scalar $z/(z+1)$ is independent of $s$ and is a unit at the
points involved in point addition, so the same argument applies to $\mathcal C$.
\end{proof}

Composing any source chart with $\mathcal H_s^{-1}$ gives a target chart.
The rational normal forms below give explicit charts
for both moduli problems and their common domain of equivalence.

For the bundles associated with the weights, the Pearson identity
$(D_w)_{22}(s)w(s)/s=w(s+1)/(s+1)$ cancels the pole at $z=s$ of
\[
 C_{s+1}[w]((z+1)^2)^{-1}D_w(z)C_s[w](z^2).
\]
Its determinant is a unit there, so $\Phi_s$ adds the squared node
$(s+1)^2$ of mass $2w(s+1)$. The same holds for $\widehat w$.
In the frames $C_s[v]$, the compatibility relation is
\begin{equation}\label{eq:positive-weight-christoffel-square}
 \begin{gathered}
 B_s^{\mathrm{res}}[v]=C_{s+1}[v]^{-1}C_s[v]
 =\begin{pmatrix}1&0\\-2v(s+1)/(y-(s+1)^2)&1\end{pmatrix},\\
 H_{s+1}^{\mathrm{res}}B_s^{\mathrm{res}}[w]
 =B_s^{\mathrm{res}}[\widehat w]H_s^{\mathrm{res}}.
 \end{gathered}
\end{equation}
Lemma~\ref{lem:christoffel-cohomology} gives its scalar normalization:
\begin{equation}\label{eq:christoffel-determinant-square}
 \frac{\varepsilon_H(s+1)}{\varepsilon_H(s)}
 =\frac{\Delta_k[\widehat w,s+1]\Delta_k[w,s]}
 {\Delta_k[\widehat w,s]\Delta_k[w,s+1]}.
\end{equation}
The supports remain disjoint at $s=k$ and $s=N-2$.
The last step ends at index $N-1$, so no connection at $N$ is needed.

\subsection{Spectral coordinates}
\label{subsec:christoffel-rational-matrix}

Put $u=z(z+1)$ and $\psi(z)=z+\tfrac12$. Define $L,M$, of degree at most
one, by the unique decomposition below, and set
\begin{equation}\label{eq:standard-paired-family}
 \begin{gathered}
 (z+1)d(z)=\psi L(u)+M(u),\qquad D_0(u)=d(z)d(-z-1),\\
 c_s(z)=(z+1)d(z)[u-s(s+1)],\qquad
 \mathscr D(u)=uD_0(u)[u-s(s+1)]^2.
 \end{gathered}
\end{equation}
\begin{lemma}[Underlying bundle of a stable connection]
\label{lem:stable-symmetric-splitting}
Every connection in $\mathcal M_s$ or $\widehat{\mathcal M}_s$
has underlying bundle isomorphic to
$\mathcal O\oplus\mathcal O(-1)$.
\end{lemma}
\begin{proof}
For a connection in $\widehat{\mathcal M}_s$, multiply its matrix
by $z/(z+1)$ as in \eqref{eq:christoffel-original-moduli}.
This restores the source finite types without changing the bundle,
its invariant line subbundles, or its leading behaviour at infinity.
By the Birkhoff--Grothendieck theorem, write
$\mathcal E\simeq\mathcal O(a)\oplus\mathcal O(-a-1)$, $a\ge0$.
In an affine frame compatible with this splitting,
the prescribed finite types make $P=c_sA$ polynomial.
In the associated holomorphic frame at infinity,
\[
 A^{(\infty)}_{21}(z)
 =z^{2a}(z+1)^{2a+2}A_{21}(z).
\]
For $m=2$, normalization at infinity gives
$A^{(\infty)}_{21}=O(z^{-1})$; for $m=1$, boundedness gives
$A^{(\infty)}_{21}=O(1)$.
Since $\deg c_s=m+3$, in both cases
\[
 P_{21}(z)
 =c_s(z)z^{-2a}(z+1)^{-2a-2}A^{(\infty)}_{21}(z)
 =O(z^{2-4a}).
\]
If $a\ge1$, this polynomial vanishes identically.
The direct summand $\mathcal O(a)$ is then generically invariant,
contrary to stability. Hence $a=0$.
\end{proof}
Use scalars $f,p,b$. The data at infinity fix $h$, with $h(f)=b$, and
the eigenline $\ell_+=\langle(q_\infty,1)^T\rangle$ by
\begin{equation}\label{eq:quadratic-plus-trace}
 \begin{gathered}
 \begin{array}{c|cc}
 m&h(u)&q_\infty\\ \hline
 2&b+(u-f)(u+f+\alpha_s)&\gamma-1\\
 1&b+\gamma(u-f)&1
 \end{array}\\
 \alpha_s=\frac{\gamma^2+[u]D_0-2s(s+1)-5/4}{2}\qquad(m=2).
 \end{gathered}
\end{equation}
On $\mathcal O\oplus\mathcal O(-1)$ the common polynomial representative is
\begin{equation}\label{eq:christoffel-input-slice}
 P=c_sA=
 \begin{pmatrix}\psi h-p&b_{12}\\
 u-f&\psi h+p\end{pmatrix},\qquad
 b_{12}=\frac{(u+\tfrac14)h^2-p^2-\mathscr D(u)}{u-f}.
\end{equation}
The condition that $b_{12}$ be polynomial is the conic equation
\begin{equation}\label{eq:quadratic-spectral-conic}
 p^2=(f+\tfrac14)b^2-\mathscr D(f).
\end{equation}

\begin{lemma}[Rational parametrization]\label{lem:quadratic-spectral-chart}
For the source and scalar-normalized target moduli problems,
the displayed normal form is unique up to scalars
when the coefficient of $u$ in the initial $P_{21}$ is nonzero.
The root $f$ of $P_{21}$ is the first coordinate. Put $r=f-s(s+1)$.
The second coordinate and reconstruction formulas are
\begin{equation}\label{eq:standard-painleve-coordinates}
 g=\frac{p-rM(f)}{b-rL(f)},
\end{equation}
\begin{equation}\label{eq:standard-painleve-inverse}
 \begin{aligned}
 b&=r\frac{L(f)(g^2+f+\tfrac14)-2M(f)g}{g^2-f-\tfrac14},\\
 p&=r\frac{2(f+\tfrac14)L(f)g-M(f)(g^2+f+\tfrac14)}{g^2-f-\tfrac14}.
 \end{aligned}
\end{equation}
This chart requires $f\ne s(s+1)$, $b\ne rL(f)$ and $g^2\ne f+1/4$
in addition to the prescribed local types.
\end{lemma}
\begin{proof}
By Lemma~\ref{lem:stable-symmetric-splitting}, the underlying bundle
is $\mathcal O\oplus\mathcal O(-1)$. Its automorphisms have constant
diagonal entries and an affine upper-right entry. The identity
$P\circ\iota=-\operatorname{adj}P$ and the bound at infinity
$\deg_z P_{21}\le2$ make $P_{21}$ affine in $u$.
Normalize it to $u-f$. The two upper-right gauge coefficients reduce
the polynomial $(P_{11}-P_{22})/2$, of degree at most two in $u$,
to its remainder $-p$ modulo $u-f$.
The trace is $2\psi h(u)$ with $\deg h\le m$,
and the determinant gives \eqref{eq:christoffel-input-slice}.
The prescribed data at infinity fix $h$ and $q_\infty$ as in
\eqref{eq:quadratic-plus-trace}. Indeed, in the holomorphic frame at infinity,
\[
 A=
 \begin{cases}
 I+\dfrac1z
 \begin{pmatrix}\tfrac12-\nu&\gamma^2-1\\1&\tfrac52-\nu\end{pmatrix}
 +O(z^{-2}),&m=2,\\[6pt]
 \begin{pmatrix}0&1\\1&0\end{pmatrix}
 +\dfrac1z\begin{pmatrix}\gamma&1-\nu\\1-\nu&\gamma\end{pmatrix}
 +O(z^{-2}),&m=1.
 \end{cases}
\]
The conic removes the pole at $u=f$.
Conversely, these formulas have the prescribed determinant and satisfy
the reflection identity and conditions at infinity. The indicated ranks
at finite points impose precisely the
remaining open local-type conditions. The normalized gauge is unique
up to scalars. Restricting to stable connections constructs the target
chart independently.
Reflecting the decomposition in \eqref{eq:standard-paired-family} gives
\[
 (f+\tfrac14)L(f)^2-M(f)^2=fD_0(f)
\]
and hence a rational point $(p,b)=r(M(f),L(f))$ on
$p^2-(f+\tfrac14)b^2=-fr^2D_0(f)$.
Intersecting the conic with a line of slope $g$ through this point gives
\eqref{eq:standard-painleve-inverse}, proving rationality.
\end{proof}

\begin{prop}[Christoffel transformation in spectral coordinates]
\label{prop:christoffel-coordinate-form}
On the common chart with $f\ne0$, form
\begin{equation}\label{eq:spectral-christoffel-frame}
 \lambda_0=-\frac{P_{22}(0)}{P_{21}(0)},\qquad
 M_0(y)=\begin{pmatrix}q_\infty+\lambda_0/y&-1\\1/y&0\end{pmatrix},\qquad
 \widetilde P=\frac z{z+1}M_0((z+1)^2)^{-1}PM_0(z^2).
\end{equation}
Define the following polynomials in $u$:
\begin{equation}\label{eq:christoffel-coefficient-update}
 C=\widetilde P_{21},\qquad
 J=\frac{\widetilde P_{11}-\widetilde P_{22}}2,\qquad
 \widetilde h=\frac{\operatorname{tr}\widetilde P}{2\psi},\qquad
 \kappa=[u]C.
\end{equation}
If $\kappa\ne0$, the target data are
\begin{equation}\label{eq:christoffel-coordinate-extraction}
 \widehat f=-C(0)/\kappa,\qquad
 \widehat p=-J(\widehat f),\qquad
 \widehat b=\widetilde h(\widehat f),\qquad \widehat\gamma=\gamma-1.
\end{equation}
Equation~\eqref{eq:standard-painleve-coordinates} gives $\widehat g$.
The resulting map is birational on the common charts.
\end{prop}
\begin{proof}
The residue of $M_0$ at zero has image $\ker P(0)$.
In holomorphic frames, its value at infinity has image $\ell_+$.
Thus $M_0$ represents the elementary modifications.
The scalar factor gives exactly $\mathcal C$ in
\eqref{eq:christoffel-original-moduli}.
Direct multiplication gives $\deg(C,J,\widetilde h)\le(1,2,2)$.
Polynomial division gives $J=(u-\widehat f)(a_1u+a_0)-\widehat p$.
The bundle automorphism
\[
 S(y)=\begin{pmatrix}1&a_1y+a_0-a_1/2\\0&\kappa\end{pmatrix}
\]
acts by $\widetilde P\mapsto S((z+1)^2)^{-1}\widetilde P S(z^2)$,
making the lower-left entry monic and half the diagonal difference constant.
At $u=\widehat f$ it leaves $\widetilde h(\widehat f)$ unchanged.
The data at infinity give $\widehat\gamma$,
and Lemma~\ref{lem:quadratic-spectral-chart} fixes all remaining coefficients.
The same construction gives a rational inverse on the common domain.
\end{proof}

The point-addition map $\Phi_s$ is rational in these coordinates. Put $a=s+1$ and
\[
 \begin{gathered}
 T_\infty(y)=\begin{pmatrix}1&q_\infty y\\0&1\end{pmatrix},\qquad
 \Pi(z)=T_\infty((z+1)^2)^{-1}P(z)T_\infty(z^2),\\
 \lambda=\frac{\Pi_{11}(s)}{\Pi_{21}(s)},\qquad
 \theta=-\left.\frac{d}{dz}
                  \frac{\Pi_{11}(z)}{\Pi_{21}(z)}\right|_{z=s}.
 \end{gathered}
\]
On the open set where $\Pi_{21}(s)\theta\ne0$, form
\[
 \begin{gathered}
 K=\frac{2a}{\theta}\binom{\lambda}{1}(-1,\lambda),\qquad
 R(y)=I-\frac{K}{y-a^2},\\
 A^+(z)=T_\infty((z+1)^2)R((z+1)^2)
       \frac{\Pi(z)}{c_s(z)}
       R(z^2)^{-1}T_\infty(z^2)^{-1}.
 \end{gathered}
\]
Put $c_{s+1}A^+$ in the normal form of
Lemma~\ref{lem:quadratic-spectral-chart}; its coordinates are
$\Phi_s(f,g)$ wherever the target chart is defined.
To verify the construction, note that $K^2=0$ and that the double zero of
$\det\Pi$ at $s$ gives
$K\Pi(s)=0$ and $K\Pi'(s)=2a\Pi(s)$.
The first identity follows from rank one. The condition $(\det\Pi)'(s)=0$
and the definition of $\theta$ give
$(-1,\lambda)\Pi'(s)=\theta(\Pi_{21}(s),\Pi_{22}(s))$, which proves the second.
Thus the pole at $s$ cancels, while the new pole lies at $s+1$;
reflection gives the other branch. This realizes
\eqref{eq:intrinsic-image-lattice}. The change is regular and invertible
on the bundle at infinity, so the data there are unchanged.
Use $\mathcal C_s$ also for the Christoffel coordinate map.
Then the transformed dynamics is
\begin{equation}\label{eq:standard-painleve-root-step}
 (\widehat f_{s+1},\widehat g_{s+1})
 =\mathcal C_{s+1}\Phi_s\mathcal C_s^{-1}(\widehat f_s,\widehat g_s).
\end{equation}
These coordinates describe the indicated open subsets of the common stable loci.

\begin{remark}[Degenerations and complex weights]
Constant ratios or coincident singularities require separate strata
of the moduli space with the corresponding local types.
The normalization identities still hold. The local formula for the
recurrence coefficient applies under its stated hypotheses.
All finite-sum identities extend to complex weights wherever the
required moment determinants are nonzero. The moduli results require
the same local types, conditions at infinity and stability.
\end{remark}
\label{sec:christoffel-end}

\section{Christoffel transformation of the gap probability recurrence}
\label{sec:christoffel-gap}
Use the positive weights and moment assumption of Section~4. Put
$\mu_\infty[v]=\sum_{x\in\mathfrak X}v(x)\delta_{x^2}$ and define
$\Delta_k[v,\infty]$ by the same moment formula, for $v=w,\widehat w$.
The following probability comparison requires no Pearson condition.

Assign each $k$-point configuration $I\subset\operatorname{supp}\mu_\infty[v]$
the weight
\[
 \prod_{x<y\in I}(y-x)^2\prod_{y\in I}\mu_\infty[v](\{y\}).
\]
After normalization, the Cauchy--Binet formula gives
\begin{equation}\label{eq:gap-tau-probability}
 \mathbb P_v(\max I\le s^2)=\mathsf E_s[v]
 =\frac{\Delta_k[v,s]}{\Delta_k[v,\infty]}.
\end{equation}
Write $\mathsf E_s=\mathsf E_s[w]$ and
$\widehat{\mathsf E}_s=\mathsf E_s[\widehat w]$.
The Christoffel transformation multiplies each configuration weight by
$\prod_{y\in I}y$. The Christoffel determinant identity of
Lemma~\ref{lem:christoffel-cohomology} gives
\begin{equation}\label{eq:gap-christoffel-scalar}
 \mathbb E_w\!\left[\prod_{y\in I}y\,\middle|\,\max I\le s^2\right]
 =(-1)^k\varepsilon_H(s)>0.
\end{equation}
Indeed, multiplication by this product replaces $\mu_s[w]$ by $y\mu_s[w]$.
The determinants $\Delta_k[w,s]$ and $\Delta_k[\widehat w,s]$ involve moments
of $w$ in $x$ only through order $4k-2$. The moment assumption gives entrywise
convergence of their matrices, and the $k$ distinct positive squared nodes
make both limiting determinants positive. Thus exhaustion of $\mathfrak X$
gives the finite nonzero limit $\varepsilon_H(\infty)$.
Dividing by the unconditional expectation yields
\begin{equation}\label{eq:gap-transform}
 \mathbb P_{\widehat w}(\max I\le s^2)
 =\mathbb P_w(\max I\le s^2)
   \frac{\varepsilon_H(s)}{\varepsilon_H(\infty)}.
\end{equation}
Consequently, a known recurrence
$\mathsf E_{s+2}\mathsf E_s=Q_s\mathsf E_{s+1}^2$ transforms into
\begin{equation}\label{eq:gap-known-recurrence}
 \widehat{\mathsf E}_{s+2}\widehat{\mathsf E}_s
 =\widehat Q_s\widehat{\mathsf E}_{s+1}^{\,2},\qquad
 \widehat Q_s=Q_s
 \frac{\varepsilon_H(s)\varepsilon_H(s+2)}{\varepsilon_H(s+1)^2}.
\end{equation}

\subsection{The recurrence coefficient in a bundle frame}

Now let $\mathfrak X=\mathbb Z$ or $\{-N,\ldots,N\}$, and retain the
positive even weight and moment assumption of Section~4. Assume that
$w(j+1)/w(j)=\varrho(j)$ on the lattice edges, where $\varrho$ is rational,
finite and positive on those edges, and
$\varrho(z)\varrho(-z-1)=1$. Use the transports
\eqref{eq:christoffel-even-connection}, with $k\le s\le N-2$ in the
finite case and $s\ge k$ otherwise. The local formula below requires no
degree bound on $\varrho$ or restriction on the eigenvalues at
infinity. The spectral description in the next subsection retains the
degree and chart assumptions of Section~4.

The degree restriction in that description concerns $w$. Indeed,
\[
 \frac{\widehat w(j+1)}{\widehat w(j)}
 =\frac{(j+1)^2}{j^2}\varrho(j)\qquad(j,j+1\ne0)
\]
need not have numerator and denominator of degree at most two.
The modification provides spectral coordinates for the transformed
connection without this degree bound.

Put $a=s+1$. For $v=w,\widehat w$, let $A_s[v]$ be its transport in
a frame of $\mathcal F_s[v]\simeq\mathcal O(-1)^{\oplus2}$.
Choose a polynomial representative $\Pi=cA_s[v]$, with $c$ having a
simple zero at $s$ and $c(a)\ne0$. The Pearson assumptions give
$\det \Pi(a)\ne0$.
Since $\Pi(s)$ has rank one, a constant change of frame ensures
$\Pi_{21}(s)\ne0$. Define
\begin{equation}\label{eq:gap-local-spectral-data}
 \begin{gathered}
 \lambda=\frac{\Pi_{11}(s)}{\Pi_{21}(s)},\qquad
 \theta=-\left.\frac{d}{dz}\frac{\Pi_{11}(z)}{\Pi_{21}(z)}\right|_{z=s},\\
 Y(z)=\lambda \Pi_{21}(z)+\Pi_{22}(z),\qquad
 F(z)=\lambda Y(z)-\lambda \Pi_{11}(z)-\Pi_{12}(z).
 \end{gathered}
\end{equation}
Primes denote $z$-derivatives. The scalar $\theta$ measures the
first-order change of the line spanned by $\Pi(z)e_1$ at $z=s$.

\begin{prop}[Formula for the recurrence coefficient]\label{prop:gap-general-pearson}
One has $\theta\ne0$, and the recurrence coefficient for either weight is
\begin{equation}\label{eq:gap-general-observable}
 \frac{\mathsf E_{s+2}[v]\mathsf E_s[v]}{\mathsf E_{s+1}[v]^2}
 =1+\frac{Y(a)F'(a)-Y'(a)F(a)}{\theta\det \Pi(a)}
 -\frac{4a(a+1)F(a)^2}{(2a+1)^2\theta^2\det \Pi(a)}.
\end{equation}
The expression is invariant under rescaling of the representative by scalar
functions regular and nonzero at $s,a$, and under constant changes of frame
where defined.
\end{prop}
\begin{proof}
We give the local normalization calculation explicitly; compare
\cite[Section~2.5]{ArinkinBorodinTau2007} for ordinary $d$-connections.
Write $\Delta_r=\Delta_k[v,r]$, $\mu_j=2v(j)$,
$\alpha=a^2$ and $\beta=(a+1)^2$.
Let $V_r(y)$ be the transpose of the normalized DRHP solution for
$\mu_r[v]$, followed by the same constant change of frame for every $r$
to obtain the frame in the statement. Its columns give a frame of
$\mathcal F_r[v]\simeq\mathcal O(-1)^{\oplus2}$, with the same
normalization at infinity for all $r$; thus
\[
 A(z)=V_s((z+1)^2)^{-1}D_v(z)V_s(z^2).
\]
Put $\mathbf u=V_s(\alpha)^{-1}e_2$ and $\ell=e_1^TV_s(\alpha)$.
The section $V_s(y)\mathbf u/(y-\alpha)$ has residue $e_2$ at $\alpha$.
Applying the new residue quotient
$\operatorname{Res}_\alpha\chi-\mu_a\phi(\alpha)$, as in
Lemma~\ref{lem:dynamics-point-addition} for the measure $\mu_s[v]$, gives
\[
 \varepsilon_s=\frac{\Delta_{s+1}}{\Delta_s}
 =1-\mu_a e_1^TV_s'(\alpha)\mathbf u>0,
 \qquad K=\frac{\mu_a}{\varepsilon_s}\mathbf u\ell.
\]
Here and below derivatives of $V_r$ are with respect to $y$.
Since $\ell\mathbf u=0$, one has $K^2=0$. The normalized section transfer gives
\[
 V_{s+1}(y)=V_s(y)\left(I+\frac{K}{y-\alpha}\right),\qquad
 R(y)=V_{s+1}(y)^{-1}V_s(y)=I-\frac{K}{y-\alpha}.
\]
The Pearson identity cancels the residue terms on the edge $s\to a$,
so $R((z+1)^2)A(z)$ is regular at $z=s$.
Writing $q(z)=\Pi(z)e_1$ and expanding there yields
\[
 Kq(s)=0,\qquad Kq'(s)=2a\,q(s).
\]
Since $q(s)\ne0$, the two vectors are independent. Hence $\theta\ne0$
and, with $e=(\lambda,1)^T$,
\[
 K=\frac{2a}{\theta}e(-1,\lambda).
\]

Choose the local branch $\sqrt\beta=a+1$, and put
\[
 t(y)=(\sqrt y-1)^2,\qquad
 B(y)=A(\sqrt y-1)^{-1}R(y)^{-1},\qquad
 \delta=(D_v)_{22}(a).
\]
The definitions give the exact identity
\[
 V_{s+1}(y)=D_v(\sqrt y-1)V_s(t(y))B(y).
\]
For both weights, the Pearson identity on the edge $a\to a+1$ gives
\[
 t'(\beta)=\frac{a}{a+1},\qquad
 \frac{\mu_{a+1}}{\delta}=\frac{\mu_a}{t'(\beta)}.
\]
Moreover, $e_1^TD_v=e_1^T$ and
$V_{s+1}(\beta)^{-1}e_2=\delta^{-1}B(\beta)^{-1}\mathbf u$. Consequently,
\[
 \begin{aligned}
 \varepsilon_{s+1}
 &=1-\mu_{a+1}e_1^TV_{s+1}'(\beta)V_{s+1}(\beta)^{-1}e_2\\
 &=\varepsilon_s-
   \frac{\mu_a}{t'(\beta)}\ell B'(\beta)B(\beta)^{-1}\mathbf u.
 \end{aligned}
\]
Dividing by $\varepsilon_s$ gives the required ratio:
\begin{equation}\label{eq:gap-local-tau-ratio}
 \begin{aligned}
 \frac{\mathsf E_{s+2}[v]\mathsf E_s[v]}{\mathsf E_{s+1}[v]^2}
 &=\frac{\varepsilon_{s+1}}{\varepsilon_s}
 =1-\frac{\operatorname{tr}(KB'(\beta)B(\beta)^{-1})}{t'(\beta)}\\
 &=1+\frac{\operatorname{tr}(KA(a)^{-1}A'(a))}{2a}
   +\frac{a+1}{a(2a+1)^2}\operatorname{tr}(KA(a)^{-1}KA(a)).
 \end{aligned}
\end{equation}
The last equality follows from
\[
 B'(\beta)B(\beta)^{-1}
 =-\frac{A(a)^{-1}A'(a)}{2(a+1)}
  -\frac{A(a)^{-1}KA(a)}{(2a+1)^2}.
\]
Here $B',t'$ are $y$-derivatives, and $A'$ is a $z$-derivative.
Since $\operatorname{tr}K=0$, replacing $A$ by $\Pi=cA$ does not
change either trace expression in \eqref{eq:gap-local-tau-ratio}.
Finally, $F=\det(e,\Pi e)$ and
$\det(\Pi e,\Pi'e)=YF'-Y'F$ imply
\[
 \operatorname{tr}(K\Pi^{-1}\Pi')
   =\frac{2a}{\theta}\frac{YF'-Y'F}{\det\Pi},\qquad
 \operatorname{tr}(K\Pi^{-1}K\Pi)
   =-\frac{4a^2}{\theta^2}\frac{F^2}{\det\Pi}.
\]
Substitution gives \eqref{eq:gap-general-observable}.
Constant changes of frame conjugate $A$ and $K$, proving the remaining
invariance assertion.
\end{proof}

\subsection{The recurrence coefficient in spectral coordinates}

On the chart of Lemma~\ref{lem:quadratic-spectral-chart}, fix $d,s,\gamma$
and reconstruct $p,b,h,b_{12}$ from $(f,g)$.
Throughout this subsection, $h,b_{12}$ are evaluated at $u=z(z+1)$.
Set $a=s+1$ and define scalar functions
\begin{equation}\label{eq:gap-spectral-scalars}
 \begin{aligned}
 \lambda&=\frac{p-\psi(s)h(s(s+1))}{f-s(s+1)}-q_\infty a^2,\qquad \theta=2q_\infty a-
 \left.\frac{d}{dz}\frac{\psi(z)h(u)-p}{u-f}\right|_{z=s},\\
 &Y(z)=\psi h+p+(\lambda+q_\infty z^2)(u-f),\\
 &F(z)=(\lambda+q_\infty(z+1)^2)Y(z)
       -(\lambda+q_\infty z^2)(\psi h-p)-b_{12}.
 \end{aligned}
\end{equation}
Define $\mathcal Q_s(f,g;\gamma)$ by the right-hand side of
\eqref{eq:gap-general-observable}, using these scalar functions and
replacing $\det\Pi(a)$ by $\mathscr D(a(a+1))$.
In these derivatives, $d,s,f,g,\gamma$ are held fixed. This expresses the recurrence
coefficient as a rational function of the spectral coordinates.

\begin{thm}[Recurrence for the transformed gap probabilities]
\label{thm:gap-christoffel}\label{thm:gap-painleve}
Where both spectral charts are defined, use the source value $\gamma$
from \eqref{eq:weight-infinity-coefficients}. Then
the recurrence coefficient for $w$ is $Q_s=\mathcal Q_s(f_s,g_s;\gamma)$.
For $(\widehat f_s,\widehat g_s)=\mathcal C_s(f_s,g_s)$, one has
\begin{equation}\label{eq:gap-painleve-recurrence}
 \widehat{\mathsf E}_{s+2}
 =\mathcal Q_s(\widehat f_s,\widehat g_s;\gamma-1)
       \frac{\widehat{\mathsf E}_{s+1}^{\,2}}{\widehat{\mathsf E}_s}.
\end{equation}
The spectral coordinates evolve by \eqref{eq:standard-painleve-root-step}.
Initial probabilities are given by
\eqref{eq:gap-transform} at $s=k,k+1$.
\end{thm}
\begin{proof}
The inclusion $\mathcal F_s[v]\subset\mathcal G_s[v]$ is the elementary
lower modification along $\ell_+$ at infinity. In these frames,
\[
 T_\infty(y)=\begin{pmatrix}1&q_\infty y\\0&1\end{pmatrix},\qquad
 \Pi=T_\infty((z+1)^2)^{-1}PT_\infty(z^2).
\]
The columns of $T_\infty$ define the splitting
$\mathcal F_s[v]\simeq\mathcal O(-1)^{\oplus2}$.
The scalar factor $z/(z+1)$ in \eqref{eq:christoffel-original-moduli}
does not affect the recurrence coefficient by Proposition~\ref{prop:gap-general-pearson}.
For $e=(\lambda,1)^T$, direct multiplication gives
\[
 \lambda=\frac{\Pi_{11}(s)}{\Pi_{21}(s)},\qquad
 \theta=-\left.(\Pi_{11}/\Pi_{21})'\right|_{s},\qquad
 Y=(\Pi e)_2,\qquad F=\det(e,\Pi e).
\]
These are \eqref{eq:gap-spectral-scalars}, with
$\det \Pi(a)=\mathscr D(a(a+1))$.
Proposition~\ref{prop:gap-general-pearson} gives the recurrence coefficient.
Section~4 gives the parameter shifts and conjugacy of the point-addition maps.
\end{proof}

\subsection{Initial data and continuation}

To compute initial data by finite sums, let $\mathsf p_j$ be the monic
orthogonal polynomials for $\mu_k[v]$, $j=k-1,k$, and set
\[
 \mathcal R_v\mathsf p(y)=\sum_{x\in\mathfrak X_k}
       \frac{v(x)\mathsf p(x^2)}{y-x^2},\qquad
 E_v(y)=\begin{pmatrix}
 \mathsf p_{k-1}&2\mathsf p_k-y\mathsf p_{k-1}\\
 \mathcal R_v\mathsf p_{k-1}&2\mathcal R_v\mathsf p_k-y\mathcal R_v\mathsf p_{k-1}
 \end{pmatrix}.
\]
Orthogonality gives
\[
 \det E_v=-2\sum_{x\in\mathfrak X_k}v(x)\mathsf p_{k-1}(x^2)^2<0.
\]
Thus $E_v$ identifies $\mathcal O\oplus\mathcal O(-1)$ with $\mathcal G_k[v]$.
Put $E_v((z+1)^2)^{-1}D_v(z)E_v(z^2)$ in the normal form of
Lemma~\ref{lem:quadratic-spectral-chart} when its conditions hold,
including the scalar factor $z/(z+1)$ for $\widehat w$.
Only moments through degree $2k-1$ in $y$ are needed.
In particular, when $s=k$ the transformed support has exactly $k$
points. The monic polynomial $\mathsf p_k$ still exists, although its norm
is zero. No formula divides by this norm.

The spectral charts need not contain all connections arising from
positive weights. Outside these charts,
\eqref{eq:gap-general-observable} in a frame of $\mathcal F_s[v]$
still gives the recurrence coefficient under the local Pearson
hypotheses above. In particular, its proof applies directly to
constant Pearson ratios and coincident eigenvalues in the data at
infinity; no spectral chart is needed.

For a reduced parameter family whose finite-cutoff moments vary regularly
(respectively, holomorphically), the expression
\[
 \frac{\Delta_k[v,s+2]\Delta_k[v,s]}{\Delta_k[v,s+1]^2}
\]
is regular (respectively, holomorphic) wherever
$\Delta_k[v,s+1]\ne0$. It agrees with the spectral expression where
the latter is defined. On each irreducible component on which the
spectral chart is dense, it is therefore the unique such extension.
The determinant expression also defines the coefficient on components
that do not meet the chart; for positive weights it is always positive.
This statement concerns the weight family, not arbitrary paths in the
$(f,g)$ plane.

At $s=N-2$, both $a^2=(N-1)^2$ and $(a+1)^2=N^2$ lie outside the
cutoff $s$. Thus $V_s$ is regular and invertible there, and the Pearson
ratio is finite and positive on the edge $N-1\to N$.
It follows that $A_s[v](a)$ is regular and invertible.
The last probability is $\mathsf E_N[v]=1$; no connection at cutoff
$N$ is required.

\label{sec:christoffel-gap-end}

\section*{Conclusion}

The Christoffel deformation and cutoff evolution act by bundle
modifications at disjoint points. Under the stated geometric hypotheses,
their compatibility gives a conjugacy of the point-addition dynamics
on the common stable loci. Moment determinants provide the normalization
needed to transfer exact gap probability recurrences along this conjugacy.
Their cohomological interpretation explains why the same determinant ratio
appears in the bundle construction and in the normalization of the
reweighted particle ensemble.

The probability comparison holds for positive even weights satisfying
the moment and support assumptions; no Pearson condition is required.
On the unit lattice, the local recurrence formula applies to rational
symmetric Pearson ratios of arbitrary degree under its local hypotheses.
Explicit rational spectral coordinates are constructed for source ratios
$d(-z-1)/d(z)$ with $\deg d=1,2$, under the stated local conditions
and conditions at infinity. They describe specified open subsets of the stable loci
and give the transformed recurrence coefficient even when the Christoffel
deformation raises the degree of the Pearson ratio.
Moment determinants, and the local formula whenever its hypotheses
hold, remain available outside these charts.

For the ensemble arising from symplectic skew Howe duality
in~\cite{NNS2026}, the identity $\lambda_1=\max_i a_i-n$ expresses
the first-row distribution in terms of the largest particle.
The construction gives exact finite-size relations for this
observable, complementing the kernel asymptotics studied there.
Iterated Christoffel transformations are a natural next step.
Deriving scaling limits would also require asymptotic control
of the recurrence solutions and their normalization factors.

\subsection*{Acknowledgements} 
The author thanks Anton Nazarov, Anton Dzhamay, Travis Scrimshaw and Viacheslav Krivorol for useful discussions. This work was supported by the Ministry of Science and Higher Education of the Russian Federation through a scientific project under agreement No. 075-15-2025-013.
\bibliographystyle{plain}
\bibliography{lit}

\end{document}